\documentclass[11pt,a4paper]{scrartcl}

\usepackage[utf8]{inputenc}
\usepackage{microtype,ellipsis}
\usepackage{lmodern}
\usepackage{amsmath}
\usepackage{amsthm}
\usepackage{amssymb}
\usepackage{mathtools}
\usepackage{paralist}
\usepackage[numbers]{natbib}
\usepackage[colorlinks, pagebackref]{hyperref}
\usepackage[dvipsnames]{xcolor}
\usepackage{tikz}
\usepackage[normal, small, width=\textwidth]{caption}
\usepackage{graphicx}
\usepackage{booktabs}
\usepackage{array}
\usepackage{multirow}
\usepackage[vlined,linesnumbered,ruled]{algorithm2e}
\usepackage{xifthen}
\usepackage{enumitem}
\usepackage{aliascnt}
\usepackage{authblk}

\usetikzlibrary{calc}
\usetikzlibrary{positioning}
\usetikzlibrary{matrix}
\usetikzlibrary{shapes}
\def\matrixDraw[#1]#2#3{
\draw[thick] (-0.2,-0.25)--(-0.3,-0.25)--(-0.3,#2*0.5+0.25)--(-0.2,#2*0.5+0.25);
\draw[thick] (#3*0.5+0.2,-0.25)--(#3*0.5+0.3,-0.25)--(#3*0.5+0.3,#2*0.5+0.25)--(#3*0.5+0.2,#2*0.5+0.25);
\foreach \i in {0,...,#2}
{
  \foreach \j in {0,...,#3}
  {\pgfmathparse{#1[\i][\j]} \let\a=\pgfmathresult
\node at (0.5*\j,0.5*\i) (a\i\j) {\a };
}
}
}

\makeatletter
\def\NAT@spacechar{~}
\makeatother

\newcommand{\dash}{\nobreakdash-\hspace{0pt}}
\newcommand{\clName}{cluster\xspace}

\newcommand{\R}{\mathbb{R}}

\newcommand{\Z}{\mathbb{Z}}
\newcommand{\N}{\mathbb{N}}
\newcommand{\I}{\mathcal{I}}
\newcommand{\J}{\mathcal{J}}
\newcommand{\calL}{\mathcal{L}}

\renewcommand*{\sectionautorefname}{Section}
\renewcommand*{\subsectionautorefname}{Section}

\newcommand{\LpCoCluProb}[1][]{%
  \ifthenelse{\isempty{#1}}%
    {\textsc{Co-Clustering}$_\calL$\xspace}
    {\textsc{Co-Clustering}$_{#1}$\xspace}
}

\newcommand{\CoClu}[1]{($#1$)-co-clustering\xspace}
\newcommand{\klpCoCluProb}[3]{$(#1,#2)$-\LpCoCluProb[#3]}
\newcommand{\matrixI}{\mathcal{A}}

\DeclareMathOperator{\cost}{cost_\infty}

\newcommand{\Sat}{\textsc{CNF-SAT}\xspace}

\newcommand{\occ}[2]{\sharp^{#1}_{#2}}

\newcommand{\problemdef}[3]{
  \vspace{10pt}\hfill
    \begin{minipage}{0.9\textwidth}%
      \textsc{#1}
      \begin{compactenum}
       \item[\textbf{Input:}] #2
    \item[\textbf{Question:}] #3
      \end{compactenum}
    \end{minipage}%
\vspace{10pt}%
}

\newcommand{\optProblemDef}[3]{
 \vspace{10pt}\hfill
    \begin{minipage}{0.9\textwidth}%
      \textsc{#1}
      \begin{compactenum}
       \item[\textbf{Input:}] #2
       \item[\textbf{Task:}] #3
      \end{compactenum}
    \end{minipage}%
\vspace{10pt}%
}

\newtheorem{Theorem}{Theorem}

\newaliascnt{lemma}{Theorem}
\newaliascnt{obs}{Theorem}
\newaliascnt{cor}{Theorem}
\newaliascnt{prop}{Theorem}
\newaliascnt{claim}{Theorem}

\newtheorem{Lemma}[lemma]{Lemma}
\newtheorem{Corollary}[cor]{Corollary}

\newtheorem{Observation}[obs]{Observation}
\newtheorem{Claim}[claim]{Claim}

\title{Co-Clustering Under the Maximum Norm\footnote{A preliminary version appeared in the proceedings of the 25th International Symposium on Algorithms and Computation (ISAAC' 14), LNCS~8889, pp.~298--309. Springer, 2014. This version contains all proofs in full detail and discusses some experimental findings.}}

\author[1]{Laurent Bulteau\thanks{Supported by the Alexander von Humboldt
Foundation, Bonn, Germany, during a postdoctural stay at TU Berlin.}}
\affil[1]{IGM-LabInfo, CNRS UMR 8049, Universit\'e Paris-Est Marne-la-Vall\'ee, France, \texttt{Laurent.Bulteau@u-pem.fr}}

\author[2]{Vincent Froese\thanks{Supported by Deutsche Forschungsgemeinschaft, project DAMM (NI 369/13).}}
\affil[2]{Institut f\"ur Softwaretechnik und Theoretische Informatik, TU Berlin, Germany, \texttt{\{vincent.froese, sepp.hartung, rolf.niedermeier\}@tu-berlin.de}}

\author[2]{Sepp Hartung}

\author[2]{Rolf Niedermeier}

\date{}

\begin{document}

\maketitle
\renewcommand{\sectionautorefname}{Section}
\renewcommand{\subsectionautorefname}{Section}

\begin{abstract}Co-clustering, that is, partitioning a numerical matrix into ``homogeneous'' submatrices,
has many applications ranging from
bioinformatics to election analysis. Many interesting variants of co-clustering are NP-hard.
We focus on the basic variant of co-clustering
where the homogeneity of a submatrix is defined in terms of minimizing the
maximum distance between two entries.
In this context, we spot several NP-hard as well as a number of relevant
polynomial-time solvable special cases, thus charting
the border of tractability for this challenging data clustering problem.
For instance, we provide polynomial-time solvability when having to
partition the rows and columns into two subsets each
(meaning that one obtains four submatrices).
When partitioning rows and columns into three subsets each, however,
we encounter NP-hardness even for input matrices containing only
values from~$\{ 0,1,2\}$.
\end{abstract}

\section{Introduction}

Co-clustering, also known as \emph{biclustering}~\cite{MO04}, performs a
simultaneous clustering of the rows and columns of a data matrix.
Roughly speaking, the problem is, given a numerical input matrix~$\matrixI$,
to partition the rows and columns of~$\matrixI$ into subsets minimizing a given
\emph{cost} function (measuring ``homogeneity'').
For a given subset~$I$ of rows and a subset~$J$ of columns, the
corresponding \emph{\clName} consists of all entries~$a_{ij}$ with $i\in I$
and~$j\in J$.
The cost function usually defines homogeneity in terms of
distances (measured in some norm) between the entries of each \clName.
Note that the variant where \clName{}s are allowed to ``overlap'',
meaning that some rows and columns are contained in multiple \clName{}s,
has also been studied~\cite{MO04}.
We focus on the non-overlapping variant which can be stated as follows.

\optProblemDef{\LpCoCluProb}
  {A matrix~$\matrixI \in \R^{m \times n}$ and two positive integers~$k,\ell \in \N$.}
  {Find a partition of $\matrixI$'s rows into~$k$ subsets and a partition of
    $\matrixI$'s columns into~$\ell$~subsets such that a given cost function
    (defined with respect to some norm~$\calL$) is minimized for the
    corresponding clustering.}
  
Co-clustering is a fundamental paradigm for
unsupervised data analysis. Its applications range from microarrays
and bioinformatics over recommender systems to election
analysis~\cite{ADK12,BDGMM07, TSS05, MO04}. Due to its enormous practical significance,
there is a vast amount of literature discussing various variants;
however, due to the observed NP-hardness of ``almost all interesting
variants''~\cite{MO04}, most of the literature deals with heuristic,
typically empirically validated algorithms.
Indeed, there has been very active research on co-clustering in terms
of heuristic algorithms while there is little substantial theoretical work for this important clustering problem.
Motivated by an effort towards a deeper theoretical analysis as started by \citet{ADK12},
we further refine and strengthen the theoretical investigations on the computational complexity of a natural special case of \LpCoCluProb, namely we study the case of~$\calL$ being the maximum norm~$L_\infty$, where the problem comes down to minimizing the
maximum distance between entries of a \clName{}.
This cost function might be a reasonable choice in practice due to its outlier sensitivity. In network security, for example, there often exists a vast amount of ``normal'' data points whereas there are only very few ``malicious'' data points, which are outliers with respect to certain attributes. 
The maximum norm does not allow to put entries with large differences into the same cluster,
which is crucial for detecting possible attacks.
The maximum norm can also be applied in a discretized setting, where input values are grouped (for example, replaced by integer values) according to their level of deviation from the mean of the respective attribute. It is then not allowed to put values of different ranges of standard deviation into the same cluster.
Last but not least, we study an even more restricted clustering version, where the partitions
of the rows and columns have to contain consecutive subsets.
This version subsumes the problem of feature discretization, which is used as a preprocessing technique in data mining applications~\cite{NS95, CN98, Ngu06}. See \autoref{sec:consec-blocks} for this version.

\citet{ADK12} provided a thorough analysis of the polynomial-time
approximability of \LpCoCluProb (with respect to~$L_p$-norms), presenting several
constant-factor approximation algorithms.
While their algorithms are almost straightforward, relying on one-dimensionally
clustering first the rows and then the columns, their main contribution
lies in the sophisticated mathematical analysis of the corresponding
approximation factors.
Note that \citet{JSB09} further generalized this approach to
higher dimensions, then called \emph{tensor clustering}.
In this work, we study (efficient) \emph{exact} instead of approximate
solvability. To this end, by focussing on \LpCoCluProb[\infty], we investigate a scenario that is combinatorially easier to grasp.
In particular, our exact and combinatorial polynomial-time algorithms exploit structural
properties of the input matrix and do not solely
depend on one-dimensional approaches.

\paragraph{Related Work.}
Our main point of reference is the work of \citet{ADK12}.
Their focus is on polynomial-time approximation
algorithms, but they also provide computational hardness results.
In particular, they point to challenging open questions concerning
the cases $k=\ell=2$, $k=1$, or binary input matrices.
Within our more restricted setting using the maximum norm,
we can resolve parts of these questions.
The survey of \citet{MO04}\footnote{According to Google Scholar, accessed December~2015,
cited more than 1500~times.}
provides an excellent overview on
the many variations of \LpCoCluProb, there called biclustering,
and discusses many applications in bioinformatics and beyond.
In particular, they also discuss Hartigan's~\cite{Har72} special case where the goal
is to partition into uniform \clName{}s
(that is, each \clName{} has only one entry value).
Our studies indeed generalize this very puristic scenario
by not demanding completely uniform \clName{}s (which would correspond to
\clName{}s with maximum entry difference~0) but allowing some variation
between maximum and minimum \clName{} entries.
\citet{CST00} aim at clusterings where in each submatrix
the distance between entries within each row and within each column is
upper-bounded.
Recent work by \citet{WUB13} considers a so-called
``monochromatic'' biclustering where the cost for each submatrix
is defined as the number of minority entries. For binary data,
this clustering task coincides with the~$L_1$-norm version of
co-clustering as defined by \citet{ADK12}.
\citet{WUB13} show NP-hardness of monochromatic biclustering for
binary data with an additional third value denoting missing entries
(which are not considered in their cost function) and give a
randomized polynomial-time approximation scheme~(PTAS).
\citet{Feige14} gave a proof of NP-hardness of \LpCoCluProb[L_1] with $k = 2$ and $\ell = n$ for binary matrices.
Note that in contrast to the $L_1$-version, we observe that the $L_\infty$-version is easily solvable on binary matrices.

\paragraph{Our Contributions.}
In terms of defining ``cluster homogeneity'',
we focus on minimizing the maximum distance between two entries within
a \clName{} (maximum norm).
\autoref{tab:results} surveys most of our results.
Our main conceptual contribution is to provide a seemingly first study
on the exact complexity of a natural special case of \LpCoCluProb[],
thus potentially stimulating a promising field
of research.

Our main technical contributions are as follows.
Concerning the computational intractability results with respect to
even strongly restricted cases, we put a lot of effort in finding the
``right'' problems to reduce from in order to make the reductions as
natural and expressive as possible, thus making non-obvious
connections to fields such as geometric set covering. Moreover,
seemingly for the first time in the context of co-clustering, we
demonstrate that the inherent NP-hardness does not stem
from the permutation combinatorics behind: the problem remains
NP-hard when all clusters must consist of consecutive rows or columns.
This is a strong constraint (the search space size is tremendously 
reduced---basically from~$k^m \cdot \ell^n$ to 
$\binom{m}{k}\cdot\binom{n}{\ell}$) which directly gives 
a polynomial-time algorithm for~$k$ and~$\ell$ being constants.
Note that in the general case we have NP-hardness for constant~$k$
and~$\ell$.
Concerning the algorithmic results, we develop a novel
reduction to SAT solving (instead of the standard reductions to
integer linear programming).
Notably, however, as opposed to previous work on polynomial-time approximation algorithms
\cite{ADK12,JSB09}, our methods seem to be tailored for the two-dimensional case
(co-clustering) and the higher dimensional case (tensor clustering)
appears to be out of reach.

\setlength{\tabcolsep}{1em}
\newcommand{\para}{$\circledast$}
\begin{table}[t!]
  \centering
  \caption{Overview of results for \klpCoCluProb{k}{\ell}{\infty}
           with respect to various parameter
           constellations ($m$: number of rows, $n$: number of columns, $|\Sigma|$:
           alphabet size, $k/\ell$: size of row/column
           partition, $c$: cost).
           A \para{} indicates that the corresponding value is considered as a parameter, where FPT means that there is an algorithm solving the problem where the superpolynomial part in the running time is a function depending solely on the parameter.
           Multiple \para{}'s indicate a combined parameterization.
           Other non-constant values may be unbounded.}
  \label{tab:results}
    \begin{tabular}{cccccl}
      \toprule
      $m$ & $|\Sigma|$ & $k$ & $\ell$ & $c$ & Complexity\\
      \midrule
      - & - & - & - & 0 & P [\autoref{thm:c=0-poly}]\\
      - & 2 & - & - & - & P [\autoref{thm:sigma=2-poly}]\\
      - & - & 1 & - & - & P [\autoref{thm:1l-poly}]\\
      - & - & 2 & 2 & - & P [\autoref{thm:22-poly}]\\
      - & - & 2 & $n$ & - & P [\autoref{thm:2n-poly}]\\
      - & 3 & 2 & - & - & P [\autoref{thm:sigma-3-poly}]\\
      - & - & 2 & \para & 1 & FPT [\autoref{cor:2l-FPT}]\\
      - & \para & 2 & - & 1 & FPT [\autoref{cor:2l-FPT}]\\
      \para & - & \para & \para & \para & FPT [\autoref{cor:FPT-mklc}]\\
      - & 3 & 3 & 3 & 1 & NP-h [\autoref{thm:33-NPhard}]\\
      2 & - & 2 & - & 2 & NP-h [\autoref{thm:2l-NPhard}]\\
      \bottomrule
    \end{tabular}
\end{table}

\section{Formal Definitions and Preliminaries}\label{sec:prelim}

We use standard terminology for matrices.
A matrix~$\matrixI=(a_{ij})\in \R^{m\times n}$
consists of~$m$~rows and $n$~columns where $a_{ij}$ denotes the
entry in row~$i$ and~column~$j$.
We define~$[n]:=\{1,2,\ldots,n\}$ and $[i,j]:=\{i,i+1,\ldots,j\}$ for
$n,i,j\in\N$.
For simplicity, we neglect running times of arithmetical
operations throughout this paper. Since we can assume that the input values of~$\matrixI$ are upper-bounded polynomially in the size~$mn$ of~$\matrixI$
(\autoref{obs:integer-values}), the blow-up in the running times
is at most polynomial.

\paragraph{Problem Definition.}
We follow the terminology of \citet{ADK12}.
For a matrix~$\matrixI\in\R^{m \times n}$, a~$(k,\ell)$-\emph{co-clustering}
is a pair~$(\I,\J)$ consisting of
a~\emph{$k$-partition}~$\I=\{I_1,\ldots,I_k\}$ of the row indices~$[m]$ of~$\matrixI$ (that is,
$I_i\subseteq [m]$ for all~$1\le i \le k$, $I_i\cap I_j=\emptyset$ for all~$1\le i < j \le k$, and~$\bigcup_{i=1}^kI_i=[m]$)
and an~$\ell$\dash{}partition~$\J=\{J_1,\ldots,J_\ell\}$ of the
column indices~$[n]$ of~$\matrixI$. We call the elements of $\I$ (resp., $\J$) row blocks (column blocks, resp.).
Additionally, we require $\I$ and~$\J$ to not contain empty sets.
For~$(r,s) \in [k]\times[\ell]$, the
set~$\matrixI_{rs}:=\{a_{ij}\in \matrixI \mid (i,j)\in I_r\times J_s\}$ is called a
\emph{\clName{}}.

The cost of a co-clustering (under maximum norm, which is the only norm we consider here)
is defined as the maximum difference between any two entries in any \clName{}, formally
$\cost(\I,\J):=\max_{(r,s)\in [k]\times[\ell]}(\max \matrixI_{rs}-\min \matrixI_{rs})$. Herein, $\max \matrixI_{rs}$ ($\min \matrixI_{rs}$) denotes the maximum (minimum, resp.) entry in~$\matrixI_{rs}$.

\begin{figure}[t]
\centering
\begin{tikzpicture}
 \def\identitymatrix{{{0,4,3,0},{2,2,1,3},{1,3,4,1}}};
 \matrixDraw[\identitymatrix]{2}{3}
 \node [left=0.1cm of a10] {$\matrixI=$};

 \begin{scope}[xshift=5cm]
    \def\identitymatrix{{{0,3,0,4},{2,1,3,2},{1,4,1,3}}};
    \matrixDraw[\identitymatrix]{2}{3}
    \node at (0.5,1.5) {$J_1$};
    \node at (3*0.5,1.5) {$J_2$};
    \node at (-0.6,2*0.5) {$I_1$};
    \node at (-0.6,0.5*0.5) {$I_2$};
    \draw (0.25+2*0.5,-0.25)--(0.25+2*0.5,1.25);
    \draw (-0.25,0.25+0.5)--(1.75,0.25+0.5);
    \node at (0.75,-1.5*0.9) {{$J_1=\{1,3,4\}$, $J_2=\{2\}$}};
    \node at (0.75,-1*0.9) {{$I_1=\{1\}$, $I_2=\{2,3\}$}};
 \end{scope}
 
 \begin{scope}[xshift=10cm]
  \def\identitymatrix{{{0,0,4,3},{1,1,3,4},{2,3,2,1}}};
  \matrixDraw[\identitymatrix]{2}{3}
  \node at (0.25,1.5) {$J_1$};
  \node at (0.25+2*0.5,1.5) {$J_2$};
  \node at (-0.6,2*0.5) {$I_1$};
  \node at (-0.6,0.5*0.5) {$I_2$};
  \draw (0.25+0.5,-0.25)--(0.25+0.5,1.25);
  \draw (-0.25,0.25+0.5)--(1.75,0.25+0.5);
  \node at (0.75,-1.5*0.9) {{$J_1=\{1,4\}$, $J_2=\{2,3\}$}};
  \node at (0.75,-1*0.9) {{$I_1=\{2\}$, $I_2=\{1,3\}$}};
 \end{scope}
\end{tikzpicture}
\caption{The example shows two~$(2,2)$-co-clusterings (middle and right) of the same matrix~$\matrixI$ (left-hand side).
It demonstrates that by sorting rows and columns according to the co-clustering, the \clName{}s can be illustrated as submatrices of this (permuted) input matrix.
The cost of the \CoClu{2,2} in the middle is three (because of the two left \clName{}s) and that of the \CoClu{2,2} on the right-hand side is one.}\label{fig:intro-example}
\end{figure}

The decision variant of \LpCoCluProb with maximum norm is as follows.

\problemdef{\LpCoCluProb[\infty]}%
{A matrix~$\matrixI \in \R^{m \times n}$, integers~$k,\ell \in \N$, and a cost~$c \ge 0$.}%
{Is there a~$(k,\ell)$-co-clustering~$(\I,\J)$ of~$\matrixI$ with~$\cost(\I,\J)\le c$?}

\noindent See \autoref{fig:intro-example} for an introductory example.
We define $\Sigma:=\{a_{ij}\in \matrixI\mid (i,j)\in [m] \times[n]\}$ to be the \emph{alphabet} of the input matrix~$\matrixI$ (consisting of the numerical values that occur in~$\matrixI$).
Note that~$|\Sigma|\le mn$.
We use the abbreviation~\klpCoCluProb{k}{\ell}{\infty}
to refer to~\LpCoCluProb[\infty] with constants $k,\ell \in\N$, and
by~\klpCoCluProb{k}{*}{\infty} we refer to the case where only~$k$ is
constant and~$\ell$ is part of the input.
Clearly, \LpCoCluProb[\infty] is symmetric with respect to~$k$
and~$\ell$ in the sense that any $(k,\ell)$-co-clustering of a
matrix~$\matrixI$ is equivalent to an $(\ell,k)$-co-clustering of the
transposed matrix~$\matrixI^T$.
Hence, we always assume that~\mbox{~$k\leq\ell$}.

We next collect some simple observations.
First, determining whether there is a cost-zero (perfect)
co-clustering is easy. Moreover, since, for a binary
alphabet, the only interesting case is a perfect co-clustering,
we get the following.
\begin{Observation}
  \label{thm:c=0-poly}\label{thm:sigma=2-poly}
  \LpCoCluProb[\infty] is solvable in $O(mn)$ time for cost zero
  and also for any size-two alphabet.
\end{Observation}
\begin{proof}
  Let $(\matrixI,k,\ell,0)$ be a \LpCoCluProb[\infty] input instance.
  For a $(k,\ell)$-co-clustering with cost 0, it holds
  that all entries of a \clName{} are equal.
  This is only possible if there are at most~$k$
  different rows and at most~$\ell$ different columns in~$\matrixI$ since
  otherwise there will be a \clName{} containing two different entries.
  Thus, the case~$c=0$ can be solved by lexicographically sorting
  the rows and columns of~$\matrixI$ in~$O(mn)$ time (e.g.\ using radix sort).
\end{proof}

We further observe that the input matrix can, without loss of
generality, be assumed to contain only integer values
(by some rescaling arguments preserving the distance relations between elements).
\begin{Observation}\label{obs:integer-values}
  For any \LpCoCluProb[\infty]-instance with arbitrary
  alphabet~$\Sigma\subset \R$, one can find in~$O(|\Sigma|^2)$ time an
  equivalent instance with alphabet~$\Sigma'\subset \Z$ and cost
  value~$c'\in \N$.
\end{Observation}

\begin{proof}
  We show that for any instance with arbitrary alphabet~$\Sigma\subset\R$ and
  cost~$c\geq0$, there exists an equivalent instance with
  $\Sigma'\subset \Z$ and $c'\in \N$.
  Let~$\sigma_i$ be the $i$-th element of~$\Sigma$ with respect to any fixed ordering.
  The idea is that the cost value~$c$ determines which elements of~$\Sigma$ are allowed to
  appear together in a cluster of a cost-$c$ co-clustering.
  Namely, in any cost-$c$ co-clustering two elements~$\sigma_i\neq \sigma_j$ can occur in the same \clName if and only if~$|\sigma_i-\sigma_j|\le c$.
  These constraints can be encoded in an undirected graph
  $G_c:=(\Sigma,E)$ with~$E:=\{\{\sigma_i,\sigma_j\}\mid \sigma_i\neq \sigma_j\in\Sigma,|\sigma_i-\sigma_j|\leq c\}$, where each vertex corresponds to an element of~$\Sigma$, and there is an edge between two vertices if and only if the corresponding elements can occur in the same cluster of a cost-$c$ co-clustering.
  
  Now, observe that~$G_c$ is a \emph{unit interval graph} since each vertex~$\sigma_i$ can be represented by the length-$c$ interval~$[\sigma_i,\sigma_i+c]$ such that it holds
  $\{\sigma_i,\sigma_j\}\in E\Leftrightarrow [\sigma_i,\sigma_i+c]\cap
  [\sigma_j,\sigma_j+c]\neq\emptyset$ (we assume all intervals to contain real values).
  By properly shifting and rescaling the intervals, one can find an embedding of~$G_c$ where the vertices~$\sigma_i$
  are represented by length-$c'$ intervals~$[\sigma'_i,\sigma'_i+c']$ of equal integer length~$c'\in\N$ with integer starting points~$\sigma'_i\in\Z$ such that $0\leq \sigma'_i \leq |\Sigma|^2$, $c'\le |\Sigma|$,
  and $|\sigma'_i-\sigma'_j|\leq c'\Leftrightarrow |\sigma_i-\sigma_j|\leq c$.
  Hence, replacing the elements~$\sigma_i$ by~$\sigma_i'$ in the input matrix yields a matrix that has a cost-$c'$ co-clustering if and only if the original input matrix has a cost-$c$ co-clustering.
  Thus, for any instance with alphabet~$\Sigma$ and cost~$c$,
  there is an equivalent instance with alphabet~$\Sigma'\subseteq \{0,\ldots,|\Sigma|^2\} $
  and cost~$c'\in \{0,\ldots,|\Sigma|\}$.
  Consequently,  we can upper-bound the values in~$\Sigma'$
  by~$|\Sigma|^2\le (mn)^2$.
\end{proof}

Due to \autoref{obs:integer-values}, we henceforth assume for the rest of the paper that the input matrix contains integers.

\paragraph{Parameterized Algorithmics.}
We briefly introduce the relevant notions from parameterized
algorithmics (refer to the monographs~\cite{Cyg15, DF13,Nie06} for a detailed introduction).
A \emph{parameterized problem}, where each instance consists of the ``classical''
problem instance~$I$ and an integer~$\rho$ called \emph{parameter}, is \emph{fixed-parameter
  tractable}~(FPT) if there is a computable function~$f$ and an
algorithm solving any instance in~$f(\rho)\cdot |I|^{O(1)}$ time.
The corresponding algorithm is called an FPT-algorithm.

\section{Intractability Results}
\label{sec:intractability}

In the previous section, we observed that \LpCoCluProb[\infty]
is easy to solve for binary input matrices (\autoref{thm:sigma=2-poly}).
In contrast to this, we show in this section
that its computational complexity significantly changes as soon as the input matrix
contains at least three different entries.
In fact, even for very restricted special cases we can show
NP-hardness. These special cases comprise co-clusterings
with a constant number of \clName{}s (\autoref{sec:const-blocks}) or input matrices with only two rows (\autoref{sec:const-rows}). 
We also show NP-hardness of finding co-clusterings where the
row and column partitions are only allowed to contain consecutive blocks~(\autoref{sec:consec-blocks}).

\subsection{Constant Number of Clusters}
\label{sec:const-blocks}

We start by showing that for input matrices containing three different
entries, \LpCoCluProb[\infty] is NP-hard even if the co-clustering
consists only of nine~\clName{}s.

\begin{Theorem}
 \label{thm:33-NPhard}
 \klpCoCluProb{3}{3}{\infty} is NP-hard for~$\Sigma =\{0,1,2\}$.
\end{Theorem}

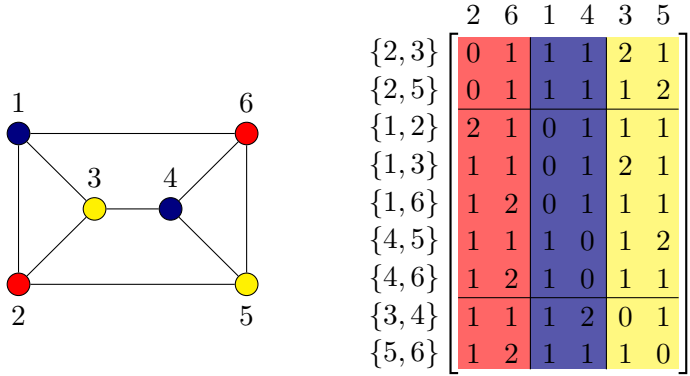
\begin{figure}[t]
  \centering
  \begin{tikzpicture}[baseline=-1.25cm]
    \tikzstyle{point}=[draw, shape=circle, inner sep=3pt];
   
    \node[point, fill=blue!50!black, label=above:{1}] (v1) at (0,2) {};
    \node[point, fill=red, label=below:{2}] (v2) at (0,0) {};
    \node[point, fill=yellow, label=above:{3}] (v3) at (1,1) {};
    \node[point, fill=blue!50!black, label=above:{4}] (v4) at (2,1) {};
    \node[point, fill=yellow, label=below:{5}] (v5) at (3,0) {};
    \node[point, fill=red, label=above:{6}] (v6) at (3,2) {};

    \draw (v1) -- (v2) -- (v3) -- (v1);
    \draw (v4) -- (v5) -- (v6) -- (v4);
    \draw (v3) -- (v4);
    \draw (v1) -- (v6);
    \draw (v2) -- (v5);
    \end{tikzpicture}
    \hspace{1cm}
    \begin{tikzpicture}
      \tikzstyle R=[fill=red, opacity=.6]
      \tikzstyle B=[fill=blue!50!black, opacity=.66]
      \tikzstyle G=[fill=yellow, opacity=.5]
      
      \fill[R] (-0.2,-0.2) rectangle (0.75,4.2);
      \fill[B] (0.75,-0.2) rectangle (1.75,4.2);
      \fill[G] (1.75,-0.2) rectangle (2.7,4.2);
     
      \def\identitymatrix{{
          {1,2,1,1,1,0},
          {1,1,1,2,0,1},
          {1,2,1,0,1,1},
          {1,1,1,0,1,2},
          {1,2,0,1,1,1},
          {1,1,0,1,2,1},
          {2,1,0,1,1,1},
          {0,1,1,1,1,2},
          {0,1,1,1,2,1}
          }};
      \matrixDraw[\identitymatrix]{8}{5}

      \node at (0,4.5) {2};
      \node at (0.5,4.5) {6};
      \node at (1.0,4.5) {1};
      \node at (1.5,4.5) {4};
      \node at (2.0,4.5) {3};
      \node at (2.5,4.5) {5};

      \node[left] at (-0.3,4.0) {$\{2,3\}$};
      \node[left] at (-0.3,3.5) {$\{2,5\}$};
      \node[left] at (-0.3,3.0) {$\{1,2\}$};
      \node[left] at (-0.3,2.5) {$\{1,3\}$};
      \node[left] at (-0.3,2.0) {$\{1,6\}$};
      \node[left] at (-0.3,1.5) {$\{4,5\}$};
      \node[left] at (-0.3,1.0) {$\{4,6\}$};
      \node[left] at (-0.3,0.5) {$\{3,4\}$};
      \node[left] at (-0.3,0.0) {$\{5,6\}$};
      
      \draw (0.25+0.5,-0.2)--(0.25+0.5,4.2);
      \draw (0.25+3*0.5,-0.2)--(0.25+3*0.5,4.2);
      \draw (-0.2,0.25+0.5)--(2.7,0.25+0.5);
      \draw (-0.2,0.25+6*0.5)--(2.7,0.25+6*0.5);
    \end{tikzpicture}
    \caption{An illustration of the reduction from~\textsc{3-Coloring}.
    Left: An undirected graph with a proper 3-coloring of the vertices such that no two neighboring vertices have the same color. Right: The corresponding matrix where the columns are labeled by vertices and the rows by edges with a $(3,3)$-co-clustering of cost~1. The coloring of the vertices determines the column partition into three columns blocks, whereas the row blocks are generated by the following simple scheme: Edges where the vertex with smaller index is red/blue (dark)/yellow (light) are in the first/second/third row block (e.g.\ the red-yellow edge~$\{2,5\}$ is in the first block, the blue-red edge~$\{1,6\}$ is in the second block, and the yellow-blue edge~$\{3,4\}$ is in the third block).}
    \label{fig:3Col}
\end{figure}

\begin{proof}
 We prove NP-hardness by reducing from the NP-complete
 \textsc{3-Coloring}~\cite{GJ79}, where the task is to partition
  the vertex set of an undirected graph into three independent sets.
  Let~$G=(V,E)$ be a \textsc{3-Coloring} instance with
  $V=\{v_1,\ldots,v_n\}$ and $E=\{e_1,\ldots,e_m\}$.
  We construct a \klpCoCluProb{3}{3}{\infty} instance $(\matrixI\in
  \{0,1,2\}^{m\times n},k:=3,\ell:=3,c:=1)$ as follows.
  The columns of~$\matrixI$ correspond to the vertices~$V$ and the rows
  correspond to the edges~$E$.
  For an edge~$e_i=\{v_j,v_{j'}\}\in E$ with $j<j'$, we set
  $a_{ij}:=0$ and $a_{ij'}:=2$. All other matrix entries are set to
  one. Hence, each row corresponding to an edge~$\{v_j,v_{j'}\}$ consists
  of 1-entries except for the columns~$j$ and~$j'$, which contain~0 and~2 (see \autoref{fig:3Col}).
  Thus, every co-clustering of~$\matrixI$ with cost at most~$c=1$ puts
  column~$j$ and column~$j'$ into different column blocks.
  We next prove that there is a $(3,3)$-co-clustering of~$\matrixI$ with
  cost at most $c=1$ if and only if $G$~admits a
  3-coloring.

  First, assume that $V_1,V_2,V_3$ is a partition of
  the vertex set~$V$ into three independent sets. We define a
  \CoClu{3,3}~$(\I,\J)$ of~$\matrixI$ as follows.
  The column partition~$\J:=\{J_1,J_2,J_3\}$ one-to-one corresponds to the
  three sets~$V_1,V_2,V_3$, that is, $J_s:=\{i\mid v_i\in V_s\}$ for all~$s\in\{1,2,3\}$.
  By the construction above, each row has exactly two non-1-entries
  being~0 and~2.
  We define the type of a row to be a permutation of~$0,1,2$, denoting which of
  the column blocks~$J_1,J_2,J_3$ contain the 0-entry and the 2-entry.
  For example, a row is of type~$(2,0,1)$ if it has a~2 in a column
  of~$J_1$ and a~0 in a column of~$J_2$.
  The row partition~$\I:=\{I_1,I_2,I_3\}$ is defined as follows:
  All rows of type~$(0,2,1)$ or~$(0,1,2)$ are put into~$I_1$.
  Rows of type~$(2,0,1)$ or~$(1,0,2)$ are contained in~$I_2$
  and the remaining rows of type~$(2,1,0)$ or~$(1,2,0)$
  are contained in~$I_3$.
  Clearly, for~$(\I,\J)$, it holds that the non-1-entries in any
  cluster are either all~0 or all~2, implying that $\cost(\I,\J)\le 1$.

  Next, assume that $(\I,\{J_1,J_2,J_3\})$ is a \CoClu{3,3}
  of~$\matrixI$ with cost at most~$1$. The vertex sets $V_1,V_2,V_3$,
  where~$V_s$ contains the vertices corresponding to the columns
  in~$J_s$, form three independent sets:
  If an edge connects two vertices in~$V_s$, then the corresponding row
  would have the 0-entry and the 2-entry in the same column
  block~$J_s$, yielding a cost of~2, which is a contradiction.
\end{proof}
\autoref{thm:33-NPhard} can even be strengthened further.
\begin{Corollary}
   \LpCoCluProb[\infty] with~$\Sigma=\{0,1,2\}$ is NP-hard for any $k
   \ge 3$, even when $\ell \geq 3$ is fixed, and
   the column blocks are forced to have equal sizes $|J_1|=\ldots =|J_\ell|$.
\end{Corollary}
\begin{proof}
Note that the reduction in \autoref{thm:33-NPhard}
clearly holds for any~$k \ge 3$.
Also, \textsc{$\ell$\dash{}Coloring} with balanced partition
sizes is still NP-hard for~$\ell \ge 3$~\cite{GJ79}.
\end{proof}

\subsection{Constant Number of Rows}
\label{sec:const-rows}

The reduction in the proof of \autoref{thm:33-NPhard} outputs matrices
with an unbounded number of rows and columns containing only
three different values.
We now show that also the ``dual restriction'' is NP-hard, that is,
the input matrix only has a constant number of rows (two) but contains
an unbounded number of different values.
Interestingly, this special case is closely related to a
two-dimensional variant of geometric set covering.

\begin{Theorem}
  \label{thm:2l-NPhard}
  \LpCoCluProb[\infty] is NP-hard for~$k=m=2$ and unbounded alphabet size~$|\Sigma|$.
\end{Theorem}

\begin{figure}[t]
  \centering
  \begin{tikzpicture}[scale=.5]
    \tikzstyle{point}=[draw, shape=circle, inner sep=1.5pt];
    \tikzstyle R=[fill=red, opacity=.6]
    \tikzstyle B=[fill=blue!50!black, opacity=.66]
    \tikzstyle G=[fill=yellow, opacity=.5]
      \draw[lightgray] (0,0) grid (5,5);
      \foreach \x in {0,1,2,3,4,5}{
        \draw (\x,0) -- (\x,-0.1);
        \node[font=\small, below] at (\x,-0.1) {\x};
        \draw (0,\x) -- (-0.1,\x);
        \node[font=\small, left] at (-0.1,\x) {\x};
      }
      \draw[R] (0,1) rectangle (2,3);
      \draw[B] (3,3) rectangle (5,5);
      \draw[G] (3,0) rectangle (5,2);
      \node[point] at (1,1) {};
      \node[point] at (0,3) {};
      \node[point] at (2,2) {};
      \node[point] at (3,4) {};
      \node[point] at (4,0) {};
      \node[point] at (5,5) {};
      \node[point] at (5,2) {};
    \end{tikzpicture}
    \hspace{1cm}
    \begin{tikzpicture}[baseline=-1.5cm]
      \tikzstyle R=[fill=red, opacity=.6]
      \tikzstyle B=[fill=blue!50!black, opacity=.66]
      \tikzstyle G=[fill=yellow, opacity=.5]
      
      \fill[R] (-0.2,-0.2) rectangle (1.25,0.7);
      \fill[B] (1.25,-0.2) rectangle (2.25,0.7);
      \fill[G] (2.25,-0.2) rectangle (3.2,0.7);
     
      \def\identitymatrix{{{3,1,2,4,5,0,2},{0,1,2,3,5,4,5}}};
      \matrixDraw[\identitymatrix]{1}{6}
      
      \draw (0.25+2*0.5,-0.2)--(0.25+2*0.5,0.7);
      \draw (0.25+4*0.5,-0.2)--(0.25+4*0.5,0.7);
      \draw (-0.2,0.25)--(3.2,0.25);
    \end{tikzpicture}
    \caption{Example of a \textsc{Box Cover} instance with seven points (left) and the corresponding \LpCoCluProb[\infty] matrix containing the coordinates of the points as columns (right). Indicated is a $(2,3)$-co-clustering of cost~2 where the column blocks are colored according to the three squares (of side length~2) that cover all points.}
    \label{fig:BoxCover}
\end{figure}
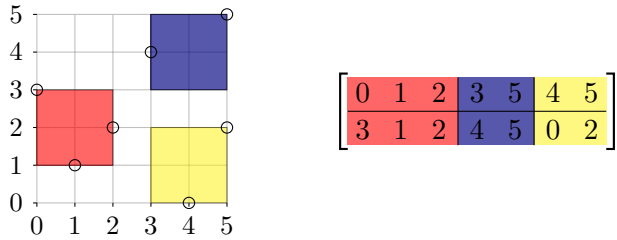

\begin{proof}
  We give a polynomial-time reduction from the NP-complete \textsc{Box
  Cover} problem~\cite{FPT81}.
  Given a set~$P\subseteq \Z^2$ of~$n$ points in the plane
  and~$\ell\in\N$, \textsc{Box Cover} is the problem to decide whether
  there are~$\ell$ squares~$S_1,\ldots,S_\ell$, each with side length~2,
    covering~$P$, that is, $P\subseteq \bigcup_{1\le s\le
    \ell} S_s$.
%

  Let~$I=(P,\ell)$ be a \textsc{Box Cover} instance. We define the
  instance~$I':=(\matrixI,k,\ell',c)$ as follows: The
  matrix~$\matrixI\in\Z^{2\times n}$ has the points~$p_1,\ldots,p_n$ in~$P$ as columns.
  Further, we set~$k:=2$, $\ell':=\ell$, $c:=2$.
  See \autoref{fig:BoxCover} for a small example.

  The correctness can be seen as follows:
  Assume that~$I$ is a yes-instance, that is, there are~$\ell$~squares~$S_1,\ldots,S_\ell$
  covering all points in~$P$. We define~$J_1:=\{j \mid p_j \in P\cap
  S_1\}$ and $J_s:=\{j\mid p_j\in P\cap S_s \setminus
  (\bigcup_{1\le l < s} S_l)\}$ for all~$2 \le s \le \ell$.
  Note that~$(\I:=\{\{1\},\{2\}\},\J:=\{J_1,\ldots,J_\ell\})$ is a~$(2,\ell)$-co-clustering
  of~$\matrixI$. Moreover, since all points with indices in~$J_s$ lie inside a square with
  side length~$2$, it holds that each pair of entries in~$\matrixI_{1s}$ as well as
  in~$\matrixI_{2s}$ has distance at most~$2$, implying~$\cost(\I,\J)\le 2$.

  Conversely, if~$I'$ is a yes-instance, then let~$(\{\{1\},\{2\}\},\J)$
  be the~$(2,\ell)$-co-clustering of cost at most~$2$.
  For any~$J_s \in \J$, it holds that all points corresponding to
  the columns in~$J_s$ have pairwise distance at most~$2$ in both
  coordinates.
  Thus, there exists a square of side length 2 covering all of them.
\end{proof}

\subsection{Clustering into Consecutive Clusters}
\label{sec:consec-blocks}

One is tempted to assume that the hardness of the previous special cases of
\LpCoCluProb[\infty] is rooted in the fact that we are allowed
to choose arbitrary subsets for the corresponding row and column
partitions since the problem remains hard even for a constant number of
\clName{}s and also with equal \clName{} sizes.
Hence, in this section, we consider a restricted version of~\LpCoCluProb[\infty],
where the row and the column partition has to consist of consecutive blocks.
Formally, for row indices $R =\{r_1,\ldots,r_{k-1}\}$ with $1 < r_1 < \ldots <
r_{k-1} \le m$ and column indices~$C =\{c_1,\ldots,c_{\ell-1}\}$
with $1 < c_1 < \ldots < c_{\ell-1} \le n$, the corresponding \emph{consecutive}
$(k,\ell)$-co-clustering~$(\I_R,\J_C)$~is~defined~as
\begin{align*}
  \I_R &:=\{\{1,\ldots,r_1-1\},\{r_1,\ldots,r_2-1\},\ldots,\{r_{k-1},\ldots,m\}\},\\
  \J_C &:=\{\{1,\ldots,c_1-1\},\{c_1,\ldots,c_2-1\},\ldots,\{c_{\ell -1},\ldots,n\}\}.
\end{align*}
The \textsc{Consecutive} \LpCoCluProb[\infty] problem now is to find a
consecutive $(k,\ell)$-co-clustering of a given input
matrix with a given cost.
Again, also this restriction is not sufficient
to overcome the inherent intractability of co-clustering, that is,
we prove it to be NP-hard.
Similarly to \autoref{sec:const-rows}, we encounter a close
relation of consecutive co-clustering to a geometric
problem, namely to find an optimal discretization of the plane; a preprocessing problem with applications in data mining~\cite{NS95, CN98, Ngu06}.
The NP-hard \textsc{Optimal Discretization} problem~\cite{CN98} is the following:
Given a set~$S=B\cup W$ of points in the plane, where each point is either colored black~($B$)
or~white~($W$), and integers~$k,\ell\in\N$, decide
whether there is a consistent set of~$k$~horizontal
and~$\ell$~vertical (axis-parallel) lines.
That is, the vertical and horizontal lines partition the plane into
rectangular regions such that no region contains two points of different colors (see \autoref{fig:example-opt-dis} for~an~example).
Here, a vertical (horizontal) line is a simple number denoting its \mbox{$x$-($y$-)coordinate.}

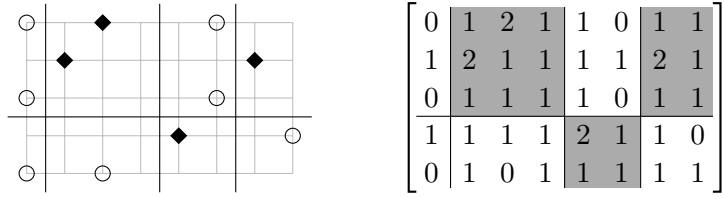
\begin{figure}[t]
\centering
\begin{tikzpicture}[scale=.5]
  \tikzstyle{pointB}=[draw, shape=diamond, fill=black, inner sep=1.5pt];
  \tikzstyle{pointR}=[draw, shape=circle, inner sep=2pt];
  \draw[lightgray] (1,1) grid (8,5);
    \foreach \x/\y in {1/1,3/1,8/2,1/3,6/3,1/5,6/5} {
      \node[pointR] at (\x ,\y) {};
    }
    \foreach \x/\y in {5/2,2/4,7/4,3/5} {
      \node[pointB] at (\x ,\y) {};
    }

     \foreach \i/\x in {1/2,2/5,3/7}{
        \draw (\x-0.5,0.5) -- (\x-0.5,5.5);
      }
       \draw (0.5,2.5) -- (8.5,2.5);
  \end{tikzpicture}
  \hspace{1cm}
  \begin{tikzpicture}
    \tikzstyle R=[fill=white]
    \tikzstyle B=[fill=black, opacity=.33]
    
    \fill[R] (-0.2,-0.2) rectangle (1.75,0.75);
    \fill[B] (1.75,-0.2) rectangle (2.75,0.75);
    \fill[R] (2.75,-0.2) rectangle (3.7,0.75);
    \fill[R] (-0.2,0.75) rectangle (0.25,2.2);
    \fill[B] (0.25,0.75) rectangle (1.75,2.2);
    \fill[R] (1.75,0.75) rectangle (2.75,2.2);
    \fill[B] (2.75,0.75) rectangle (3.7,2.2);

    \def\identitymatrix{{
        {0,1,0,1,1,1,1,1},
        {1,1,1,1,2,1,1,0},
        {0,1,1,1,1,0,1,1},
        {1,2,1,1,1,1,2,1},
        {0,1,2,1,1,0,1,1}}};
    \matrixDraw[\identitymatrix]{4}{7}
    
    \draw (0.25,-0.2) -- (0.25,2.2);
    \draw (1.75,-0.2) -- (1.75,2.2);
    \draw (2.75,-0.2) -- (2.75,2.2);
    \draw (-0.2,0.75) -- (3.7,0.75);
    \end{tikzpicture}
\caption{Example instance of \textsc{Optimal Discretization} (left)
  and the corresponding instance of \textsc{Consecutive} \LpCoCluProb[\infty] (right).
  The point set consists of white (circles) and black (diamonds) points.
A solution for the corresponding \textsc{Consecutive} \LpCoCluProb[\infty] instance
(shaded \clName{}s) naturally translates into a consistent set of lines.}\label{fig:example-opt-dis}\vspace{-5pt}
\end{figure}

\begin{Theorem}\label{thm:cons-NPhard}
  \textsc{Consecutive} \LpCoCluProb[\infty] is NP-hard for~$\Sigma =\{0,1,2\}$.
\end{Theorem}

\begin{proof}
  We give a polynomial-time reduction from \textsc{Optimal Discretization}.
  Let~$(S,k,\ell)$ be an \textsc{Optimal Discretization} instance and
  let~$X := \{x_1^*,\ldots,x_n^*\}$ be the set of different
  $x$-coordinates and let~$Y :=\{y_1^*,\ldots,y_m^*\}$ be the set of
  different~$y$-coordinates of the points in~$S$. Note that~$n$
  and~$m$ can be smaller than~$|S|$ since two points can have the
  same~$x$- or $y$-coordinate. Furthermore, assume that~$x_1^*  <
  \ldots < x_n^*$ and $y_1^* < \ldots < y_m^*$.
  We now define the  \textsc{Consecutive} \LpCoCluProb[\infty] instance~$(\matrixI,k+1,\ell+1,c)$ as
  follows:
  The matrix~$\matrixI\in\{0,1,2\}^{m\times n}$ has columns labeled with
  $x_1^*,\ldots,x_n^*$ and rows labeled with $y_1^*,\ldots, y_m^*$.
  For $(x,y)\in X\times Y$, the entry~$a_{xy}$ is defined as~$0$~if $(x,y)\in W$,~2~if $(x,y)\in B$, and otherwise~1.
  The cost is set to~$c:=1$.
  Clearly, this instance can be constructed in polynomial time.  
  
  To verify the correctness of the reduction, assume first that~$I$ is a
  yes-instance, that is, there is a set~$H=\{x_1,\ldots,x_k\}$ of~$k$ horizontal
  lines and a set~$V=\{y_1,\ldots,y_\ell\}$ of~$\ell$ vertical lines
  partitioning the plane consistently. We define row
  indices~$R:=\{r_1,\ldots,r_k\}$, $r_i := \max\{x^*\in X \mid x^* \le
  x_i\}$ and column indices~$C:=\{c_1,\ldots,c_\ell\}$, $c_j:=\max\{y^*\in Y \mid y^*
  \le y_j\}$. For the corresponding~$(k+1,\ell+1)$-co-clustering~$(\I_R,\J_C)$, it holds
  that no \clName{} contains both values~$0$ and~$2$, since otherwise
  the corresponding partition of the plane defined by~$H$ and~$V$
  contains a region with two points of different colors, which
  contradicts consistency.
  Thus, we have~$\cost(\I_R,\J_C)\le 1$, implying that~$I'$ is a
  yes-instance.

  Conversely, if~$I'$ is a yes-instance, then there exists
  a~$(k+1,\ell+1)$-co-clustering~$(\I_R,\J_C)$ with cost at most~$1$,
  that is, no \clName{} contains both values~$0$ and~$2$.
  Clearly, then the~$k$ horizontal lines~$x_i :=
  \min I_{i+1}$, $i = 1,\ldots,k$, and the~$\ell$ vertical lines~$y_j:=\min
  J_{j+1}$, $j=1,\ldots,\ell$, are consistent. Hence, $I$ is a yes-instance.
\end{proof}
Note that even though  \textsc{Consecutive} \LpCoCluProb[\infty] is
NP-hard, there still is some difference in its computational complexity
compared to the general version.
In contrast to \LpCoCluProb[\infty], the consecutive version
is polynomial-time solvable for constants~$k$ and~$\ell$ by simply trying out
all~$O(m^kn^\ell)$ consecutive partitions of the rows~and~columns.

\section{Tractability Results}

In \autoref{sec:intractability}, we showed that~\LpCoCluProb[\infty] is NP-hard
for~$k = \ell = 3$ and also for~$k=2$ in case of
unbounded~$\ell$ and~$|\Sigma|$.
In contrast to these hardness results, we now investigate
which parameter combinations yield tractable cases.
It turns out (\autoref{sec:poly-time}) that the problem is polynomial-time solvable
for~$k=\ell=2$ and for~$k=1$. We can even solve the case~$k=2$
and~$\ell \ge 3$ for~$|\Sigma|=3$ in polynomial time
by showing that this case is in fact equivalent
to the case~$k=\ell=2$.
Note that these tractability results nicely complement the hardness results
from \autoref{sec:intractability}.
We further show fixed-parameter tractability for the parameters size of the alphabet~$|\Sigma|$ and the number of column blocks~$\ell$~(\autoref{sec:FPT}).

We start (\autoref{sec:sat}) by describing a reduction of \LpCoCluProb[\infty] to \Sat (the satisfiability problem for Boolean formulas in conjunctive normal form).
Later on, it will be used in some special cases (see \autoref{thm:22-poly} and \autoref{thm:2l-cost1-FPT})
because there the corresponding formula---or an equivalent formula---only
consists of clauses containing two literals, thus being a
polynomial-time solvable \textsc{2\dash{}SAT} instance.

\subsection{Reduction to CNF-SAT Solving}\label{sec:sat}

In this section we describe two approaches to solve \LpCoCluProb[\infty] via \Sat.
The first approach is based on a straightforward reduction of a \LpCoCluProb[\infty] instance to one \Sat instance with clauses of size at least four. Note that this does not yield any theoretical improvements in general.
Hence, we develop a second approach which requires to solve~$O(|\Sigma|^{k\ell})$ many \Sat instances
with clauses of size at most~$\max\{k,\ell,2\}$. The theoretical advantage of this approach is that if~$k$ and~$\ell$ are constants, then there are only polynomially many \Sat instances to solve.
Moreover, the CNF formulas contain smaller clauses (for~$k\le\ell\le 2$, we even obtain polynomial-time solvable \textsc{2\dash{}SAT} instances).
While the second approach leads to (theoretically) tractable special cases, it is not clear that it also performs better in practice. This is why we conducted some experiments for empirical comparison of the two approaches (in fact, it turns out that the straightforward approach allows to solve larger instances). In the following, we describe the reductions in detail and briefly discuss the experimental results.

We start with the straightforward polynomial-time reduction from \LpCoCluProb[\infty] to \Sat.
We simply introduce a variable~$x_{i,r}$ ($y_{j,s}$) for each pair of row index~$i\in[m]$ and row
block index~$r\in[k]$ (respectively column index~$j\in[n]$ and column block index~$s\in[\ell]$)
denoting whether the respective row (column) may be put into the
respective row (column) block. For each row~$i$, we enforce that it is
put into at least one row block with the clause~$(x_{i,1}\vee \ldots \vee x_{i,k})$ (analogously for the
columns). We encode the cost constraints by introducing $k\ell$
clauses~$(\neg x_{i,r}\vee \neg x_{i',r} \vee \neg y_{j,s} \vee \neg
y_{j',s})$, $(r,s)\in[k]\times[\ell]$ for each pair of
entries~$a_{ij}, a_{i'j'}\in \matrixI$ with~$|a_{ij}-a_{i'j'}|> c$.
These clauses simply ensure that~$a_{ij}$ and~$a_{i'j'}$ are not put
into the same \clName.
Note that this reduction yields a \Sat instance with~$km+\ell n$
variables and~$O((mn)^2k\ell)$ clauses of size up to~$\max\{k,\ell,4\}$.

Based on experiments\footnote{Using the PicoSAT Solver of
  \citet{Biere08}.}, which we conducted on randomly generated synthetical data (of size up to $m=n=1000$)
as well as on a real-world data set\footnote{Animals with Attributes dataset (http://attributes.kyb.tuebingen.mpg.de).} (with~$m=50$ and~$n=85$),
we found that we can solve
instances up to~$k=\ell=11$ using the above \Sat approach.
In our experiments we first computed an upper and a lower bound on the optimal cost value~$c$ and then created the \Sat instances  for decreasing values for~$c$, starting from the upper bound. 
The upper and the lower bound have been obtained as follows:
Given a \klpCoCluProb{k}{\ell}{\infty} instance on~$\matrixI$, solve \klpCoCluProb{k}{n}{\infty}
and~\klpCoCluProb{m}{\ell}{\infty} separately for input matrix~$\matrixI$.
Let~$(\I_1,\J_1)$ and~$(\I_2,\J_2)$ denote the~$(k,n)$-
and~$(m,\ell)$-co-clustering respectively, and let their costs be~$c_1:=\text{cost}(\I_1,\J_1)$
and~$c_2:=\text{cost}(\I_2,\J_2)$. 
We take $\max\{c_1,c_2\}$ as a lower bound and $c_1+c_2$ as an upper bound on the optimal cost value for an optimal \CoClu{k,\ell} of~$\matrixI$.
It is straightforward to argue on the correctness of the lower bound and we next show that $c_1+c_2$ is an upper bound.
Consider any pair~$(i,j), (i',j')\in[m]\times [n]$ such that~$i$
and~$i'$ are in the same row block of~$\I_1$, and~$j$ and~$j'$ are in
the same column block of~$\J_2$ (that is,~$a_{ij}$ and~$a_{i'j'}$ are
in the same \clName). Then, it holds $|a_{ij}-a_{i'j'}| \le
|a_{ij}-a_{i'j}|+|a_{i'j}-a_{i'j'}| \le c_1 + c_2$.
Hence, just taking the row partitions from $(\I_1,\J_1)$ and the column partitions from ~$(\I_2,\J_2)$ gives a combined \CoClu{k,\ell} of cost at most~$c_1+c_2$.

From a theoretical perspective, the above naive approach of solving
\LpCoCluProb[\infty] via \Sat does not yield any improvement in terms
of polynomial-time solvability.
Therefore, we now describe a different approach which leads to some
polynomial-time solvable special cases.
To this end, we introduce the concept of
\emph{\clName{} boundaries}, which are basically
lower and upper bounds for the values in a \clName{} of a co-clustering.
Formally, given two integers $k,\ell$, an alphabet~$\Sigma$, and a cost~$c$,
we define a \clName{} boundary to be a matrix $\mathcal U=(u_{rs}) \in
\Sigma^{k\times \ell}$.
We say that a $(k,\ell)$-co-clustering of~$\matrixI$ \emph{satisfies} a \clName{} boundary~$\mathcal U$
if $\matrixI_{rs}\subseteq [u_{rs},u_{rs}+c]$ for all $(r,s)\in[k]\times[\ell]$.
It can easily be seen that a given $(k,\ell)$-co-clustering has cost at most~$c$ if and only if it satisfies
at least one \clName{} boundary~$(u_{rs})$, namely, the one with $u_{rs}=\min \matrixI_{rs}$.

The following ``subtask'' of \LpCoCluProb[\infty] can be reduced to a certain \Sat instance:
Given a \clName{} boundary~$\mathcal{U}$ and a
\LpCoCluProb[\infty] instance~$I$, find a co-clustering for~$I$ that
satisfies~$\mathcal{U}$.
The polynomial-time reduction provided by the following lemma can be used to obtain exact \LpCoCluProb[\infty] solutions with the
help of SAT solvers and we use it in our subsequent~algorithms.

\begin{Lemma}\label{lem:reduc-to-sat}
  Given a \LpCoCluProb[\infty]-instance~$(\matrixI,k,\ell,c)$ and a \clName{}
  boundary~$\mathcal{U}$, one can construct in polynomial time
  a \Sat instance~$\phi$ with at most $\max\{k,\ell, 2\}$ variables per
  clause such that $\phi$~is satisfiable if and only if there is a
  \CoClu{k,\ell} of~$\matrixI$ which satisfies~$\mathcal{U}$.
\end{Lemma}
\begin{proof}
  Given an instance $(\matrixI,k,l,c)$ of \LpCoCluProb[\infty] and a \clName{}
  boundary $\mathcal U=(u_{rs})\in\Sigma^{k\times\ell}$, we define the following Boolean variables:
  For each $(i,r)\in[m]\times[k]$, the variable $x_{i,r}$
  represents the expression ``row $i$ could be put into row block $I_r$''.
  Similarly, for each $(j,s)\in[n]\times[\ell]$, the variable
  $y_{j,s}$ represents that ``column $j$ could be put into column block $J_s$''.

  We now define a Boolean CNF formula~$\phi_{\matrixI,\mathcal{U}}$ containing the
  following clauses: A clause $R_i := (x_{i,1}\vee
  x_{i,2} \vee \ldots \vee x_{i,k})$ for each row $i\in[m]$ and a clause
  $C_j:=(y_{j,1}\vee y_{j,2} \vee \ldots \vee y_{j,\ell})$ for each column~$j\in [n]$.
  Additionally, for each $(i,j)\in [m]\times [n]$ and each
  $(r,s)\in[k]\times[\ell]$ such that element $a_{ij}$ does not fit into the
  \clName{} boundary at coordinate~$(r,s)$, that is, $a_{ij}\notin
  [u_{rs},u_{rs}+c]$, there is a clause $ B_{ijrs}:=(\neg x_{i,r} \vee \neg
  y_{j,s}$).
  Note that the clauses~$R_i$ and~$C_j$ ensure that row~$i$ and
  column~$j$ are put into some row and some column block respectively.
  The clause~$B_{ijrs}$ expresses that it is impossible to have both
  row~$i$ in block~$I_r$ and column~$j$ in block~$J_s$ if $a_{ij}$
  does not satisfy $u_{rs}\leq a_{ij}\leq u_{rs}+c$.
  Clearly,~$\phi_{\matrixI,\mathcal{U}}$~is satisfiable if and only if there exists a
  \CoClu{k,\ell} of~$\matrixI$ satisfying the \clName{} boundary $\mathcal U$.
  Note that $\phi_{\matrixI,\mathcal U}$ consists of $km+\ell n$ variables
  and $O(mnk\ell)$ clauses.
\end{proof}

Using \autoref{lem:reduc-to-sat}, we can solve \LpCoCluProb[\infty]
by solving~$O(|\Sigma|^{k\ell})$ many \Sat instances (one for each possible cluster boundary) with~$km+\ell n$
variables and~$O(mnk\ell)$ clauses of size at most~$\max\{k,\ell,2\}$.
We also implemented\footnote{Python scripts available at http://www.akt.tu-berlin.de/menue/software.} this approach for comparison with the straightfoward reduction to \Sat above.
The bottleneck of this approach, however, is the number of possible cluster boundaries, which grows extremely fast.
While a single \Sat instance can be solved quickly, generating all possible cluster boundaries together with the corresponding CNF formulas becomes quite expensive, such that we could only solve instances with very small values of~$|\Sigma|\le 4$ and~$k\le\ell\le 5$.

\subsection{Polynomial-Time Solvability}
\label{sec:poly-time}

We first present a simple and efficient algorithm for \klpCoCluProb{1}{*}{\infty},
that is, the variant where all rows belong to one row block.

\begin{Theorem}
  \label{thm:1l-poly}
  \klpCoCluProb{1}{*}{\infty} is solvable in $O(n(m+\log n))$ time.
\end{Theorem}
\begin{proof}

  \begin{algorithm}[t]
    \caption{Algorithm for  \klpCoCluProb{1}{*}{\infty} }
    \KwIn{$\matrixI\in \R^{m\times n}$, $\ell\geq 1$, $c\geq 0$.}
    \KwOut{A partition of $[n]$ into at most $\ell$ blocks yielding a
      cost of at most $c$, or \texttt{no} if no such partition exists.}
    \label{algo:k1-infty}
    \For{$j\gets 1$ to $n$}{
      $\alpha_j \gets \min \{a_{ij}\mid 1\leq i\leq m\}$\;
      $\beta_j \gets \max \{a_{ij}\mid 1\leq i\leq m\}$\;
      \If { $\beta_j - \alpha_j > c$}{
        \Return{\texttt{no}};\label{l:trivialNo}
      }
    }
    $\mathcal N\gets[n]$\;
    \For{$s \gets 1 $ to $\ell$}  {
      Let $j_s\in \mathcal N$ be the index such that $\alpha_{j_s}$ is minimal\;
      $J_s\gets \{j\in \mathcal N \mid \beta_j-\alpha_{j_s} \leq c\}$\;
      $\mathcal N\gets \mathcal N \setminus  J_s$\;
      \If { $\mathcal N=\emptyset $ }{
        \Return {$(J_1,\ldots,J_s)$}\label{l:returnSol}\;
      }
    }
    \Return{\texttt{no}}\label{l:returnFalse2}\;
  \end{algorithm}

  We show that \autoref{algo:k1-infty} solves \klpCoCluProb{1}{*}{\infty}.
  In fact, it even computes the minimum~$\ell'$ such that~$\matrixI$ has
  a~$(1,\ell')$-co-clustering of cost~$c$.
  The overall idea is that with only one row block all entries of a column~$j$
  are contained in a cluster in any solution, and thus, it suffices to consider
  only the minimum~$\alpha_j$ and the maximum~$\beta_j$ value in column~$j$.
  More precisely, for a column block~$J\subseteq[n]$ of a solution it follows that~$\max\{\beta_j \mid j\in J\} - \min\{\alpha_j\mid j\in J\} \le c$.
  The algorithm starts with the column~$j_1$ that contains the overall minimum value~$\alpha_{j_1}$ of the input matrix, that is, $\alpha_{j_1}=\min\{\alpha_j\mid j\in[n]\}$. Clearly, $j_1$ has to be contained in some column block, say~$J_1$. The algorithm then adds all other columns~$j$ to~$J_1$ where~$\beta_j \le \alpha_{j_1}+c$, removes the columns~$J_1$ from the matrix, and recursively proceeds with
  the column containing the minimum value of the remaining matrix.
  We continue with the correctness of the described procedure.
  
  If \autoref{algo:k1-infty} returns $(J_1,\ldots,J_{\ell'})$ at
  Line~\ref{l:returnSol}, then this is a column partition into ${\ell'}\leq \ell$ blocks
  satisfying the cost constraint. First, it is a partition by construction:
  The sets~$J_s$ are successively removed from~$\mathcal N$ until it is empty.
  Now, let $s\in [\ell']$.
  Then, for all $j\in J_s$, it holds $\alpha_j\geq \alpha_{j_s}$ (by definition
  of~$j_s$) and $\beta_j\leq \alpha_{j_s}+c$ (by definition of~$J_s$).
  Thus, $\matrixI_{1s}\subseteq [\alpha_{j_s},\alpha_{j_s}+c]$ holds for all~$s\in[\ell']$, which
  yields~$\cost(\{[m]\},\{J_1,\ldots,J_{\ell'}\}) \le c$.

  Otherwise, if \autoref{algo:k1-infty} returns \texttt{no} in Line~\ref{l:trivialNo}, then
  it is clearly a no-instance since the difference between the maximum and the minimum value in a column is larger than~$c$.
  If \texttt{no} is returned in Line~\ref{l:returnFalse2},
  then the algorithm has computed column indices~$j_s$ and column blocks~$J_s$ for each $s\in[\ell]$, and there still exists at
  least one index~$j_{\ell+1}$ in $\mathcal N$ when the algorithm terminates.
  We claim that the columns~$j_1,\ldots,j_{\ell+1}$ all have to be in different blocks
  in any solution.
  To see this, consider any~$s,s'\in [\ell+1]$ with~$s < s'$.
  By construction, $j_{s'}\notin J_s$.
  Therefore, $\beta_{j_{s'}}>\alpha_{j_s}+c$ holds, and columns~$j_s$
  and~$j_{s'}$ contain elements with distance more than~$c$.
  Thus, in any co-clustering with cost at most~$c$, columns $j_1,\ldots,
  j_{\ell+1}$ must be in different blocks, which is impossible with only~$\ell$ blocks.
  Hence, we indeed have a no-instance.

  The time complexity is seen as follows. The first loop examines in~$O(mn)$ time
  all elements of the matrix.
  The second loop can be performed in $O(n)$ time if the $\alpha_j$ and the
  $\beta_j$ are sorted beforehand, requiring $O(n\log n)$ time.
  Overall, the running time is in $O(n(m+ \log n))$.
\end{proof}
From now on, we focus on the $k=2$ case, that is, we need to partition the
rows into two blocks. We first consider the simplest case, where also $\ell=2$.

\begin{Theorem}
 \label{thm:22-poly}
 \klpCoCluProb{2}{2}{\infty} is solvable in $O(|\Sigma|^2mn)$ time.
\end{Theorem}

\begin{proof}
  We use the reduction to \Sat provided by \autoref{lem:reduc-to-sat}.
  First, note that a \clName{} boundary $\mathcal U\in\Sigma^{2\times 2}$ can only be
  satisfied if it contains the elements~$\min\Sigma$ and
  $\min\{a\in\Sigma\mid a\ge \max\Sigma -c\}$.
  The algorithm enumerates all $O(|\Sigma|^2)$ of these \clName{} boundaries.
  For a fixed $\mathcal U$, we construct the Boolean formula
  $\phi_{\matrixI,\mathcal U}$. Observe that this formula is in 2-CNF form:
  The formula consists of $k$-clauses, $\ell$-clauses, and
  $2$-clauses, and we have $k=\ell=2$.
  Hence, we can determine whether it is satisfiable in linear time \cite{APT79} (note that the size of the formula is in $O(mn)$).
  Overall, the input is a yes-instance if and only if
  $\phi_{\matrixI,\mathcal U}$ is satisfiable for some \clName{} boundary~$\mathcal U$.
\end{proof}

Next, we consider the case $\ell=n$, that is, every column block contains a single column (also called ``one-sided'' clustering).
We show that this case is polynomial-time solvable in contrast to the NP-hardness under the~$L_1$-norm~\cite{Feige14} and also under the~$L_2$-norm (\klpCoCluProb{k}{*}{2} with $\ell=n$ is the \textsc{$k$-Means} problem)~\cite{ADHP09}.

\begin{Theorem}\label{thm:2n-poly}
  \klpCoCluProb{2}{*}{\infty} can be solved in $O(nm^2)$ time for~$\ell=n$.
\end{Theorem}

\begin{proof}
  Let $I=(\matrixI\in \Sigma^{m\times n},k=2,\ell=n,c)$ be a
  \klpCoCluProb{2}{*}{\infty} instance. Clearly, we can assume that a \CoClu{2,n}~$(\I, \J)$ has the column partition~$\J = \{\{1\},\ldots,\{n\}\}$. It remains to find a 2-partition of the rows that yields cost at most~$c$. This is simply a \textsc{2-Coloring} problem on the graph where the rows are the vertices and there is an edge between two rows~$i$ and~$i'$ if there exists a column~$j$ such that $|a_{ij} - a_{i'j}| > c$. Clearly, rows~$i$ and~$i'$ can only be in the same row block if there is no edge between them. We can solve \textsc{2-Coloring} in time linear in the size of the graph which is in~$O(m^2)$. Constructing the graph requires~$O(nm^2)$ time.
\end{proof}

Finally, we show linear-time solvability for any number of column blocks on trinary matrices.
\begin{Theorem}\label{thm:sigma-3-poly}
  \klpCoCluProb{2}{*}{\infty} is $O(mn)$-time solvable for~$|\Sigma|=3$.
\end{Theorem}

\begin{proof}
  Let $I=(\matrixI\in \{\alpha,\beta,\gamma\}^{m\times n},k=2,\ell,c)$ be a
  \klpCoCluProb{2}{*}{\infty} instance.
  We assume without loss of generality that $\alpha<\beta<\gamma$.
  The case $\ell \le 2$ is solvable in~$O(mn)$ time by \autoref{thm:22-poly}.
  Hence, it remains to consider the case~$\ell \ge 3$.
  As~$|\Sigma|=3$, there are four potential values for a minimum-cost \CoClu{2,\ell}.
  Namely, cost~0 (all \clName{} entries are equal), cost~$\beta-\alpha$,
  cost~$\gamma - \beta$, and cost~$\gamma-\alpha$.
  Since every \CoClu{2,\ell} is of cost at most~$\gamma-\alpha$ and
  because it can be checked in~$O(mn)$ time whether there is a
  \CoClu{2,\ell} of cost~0 (\autoref{thm:c=0-poly}),
  it remains to check whether there is a \CoClu{2,\ell} for~$c\in\{\beta-\alpha,\gamma - \beta\}$.

  \emph{Avoiding} a pair $(x,y) \in \{\alpha,\beta,\gamma\}^2$ means to
  find a co-clustering without a \clName{} containing~$x$ and~$y$.
  If~$c=\max\{\beta-\alpha,\gamma - \beta\}$ (Case~1), then the problem comes down
  to finding a \CoClu{2,\ell} avoiding the pair $(\alpha, \gamma)$.
  Otherwise~(Case~2), the problem is to find a \CoClu{2,\ell}
  avoiding the pair~$(\alpha,\gamma)$ and, additionally,
  either~$(\alpha,\beta)$ or~$(\beta,\gamma)$.

  \textbf{Case~1.} Finding a \CoClu{2,\ell} avoiding $(\alpha,\gamma)$: \\
  In this case, we substitute $\alpha:=0$, $\beta:=1$, and $\gamma:=2$ and set~$c:=1$.
  First, we can assume that $\ell\le 4$. To see this, note that
  for every cluster $\matrixI_{rs}$, $(r,s)\in[2]\times[\ell]$, of a cost-1 \CoClu{2,\ell}, either
  $\matrixI_{rs}\subseteq\{0,1\}$ or $\matrixI_{rs}\subseteq\{1,2\}$ holds.
  This allows us to define four \emph{types} of column blocks:
  Column block~$J_s$, $s\in[\ell]$, has type
  \begin{itemize}
    \item $(0,0)$ if $\matrixI_{1s}\subseteq\{0,1\}$ and $\matrixI_{2s}\subseteq\{0,1\}$,
    \item $(0,2)$ if $\matrixI_{1s}\subseteq\{0,1\}$ and $\matrixI_{2s}\subseteq\{1,2\}$,
    \item $(2,0)$ if $\matrixI_{1s}\subseteq\{1,2\}$ and $\matrixI_{2s}\subseteq\{0,1\}$, and
    \item $(2,2)$ if $\matrixI_{1s}\subseteq\{1,2\}$ and $\matrixI_{2s}\subseteq\{1,2\}$.
  \end{itemize}
  Clearly, any two column blocks with the same type can be merged.
  Hence, there also exists a \CoClu{2,4} of cost~1.
  
  We now describe how to find a \CoClu{2,\ell} of cost~1 for~$\ell\in\{3,4\}$.
  First, note that all columns that contain only values in~$\{0,1\}$ or only values in~$\{1,2\}$ do not restrict the row partition at all. 

  For $\ell=4$, we start to put all columns without a~2 into a type-(0,0) column block and all columns without a~0 into a type-(2,2) column block. This is correct since these column blocks always yield clusters of cost at most~1 regardless of the row blocks.
  Note that all remaining columns have to end up in column blocks of type (2,0) or (0,2) (that is, they require at most two column blocks).
  Hence, there exists a cost-1 \CoClu{2,4} if and only of if there exists a cost-1 \CoClu{2,2} for the remaining columns.
  This can be checked in $O(mn)$ time (\autoref{thm:22-poly}).

  If $\ell=3$, then analogous arguments as above apply except that we have two options: either removing a type-(0,0) column block first or removing a type-(2,2) column block first. In both cases, we then check whether the remaining matrix has a cost-1 \CoClu{2,2} in~$O(mn)$ time (\autoref{thm:22-poly}).
  
  \textbf{Case~2:} Finding a \CoClu{2,\ell} avoiding $(\alpha,\gamma)$
  and $(\alpha,\beta)$  (or $(\beta,\gamma)$):\\
  In this case, we substitute $\alpha:=0$, $\gamma:=1$,
  and~$\beta:=1$ if $(\alpha,\beta)$ has to be avoided,
  or~$\beta:=0$ if $(\beta,\gamma)$ has to be avoided.
  It remains to determine whether there is a \CoClu{2,\ell}
  with cost~0, which can be done in~$O(mn)$ time due to
  \autoref{thm:c=0-poly}.
\end{proof}

\subsection{Fixed-Parameter Tractability}\label{sec:FPT}

We develop an algorithm solving \klpCoCluProb{2}{*}{\infty} for $c=1$
based on our reduction to \Sat (see \autoref{lem:reduc-to-sat}).
The main idea is, given matrix~$\matrixI$ and \clName{} boundary~$\mathcal U$,
to simplify the Boolean formula $\phi_{\matrixI,\mathcal U}$ into a \textsc{2\dash{}Sat} formula which
can be solved efficiently.
This is made possible by the constraint on the cost, which imposes a very specific
structure on the \clName{} boundary.
This approach requires to enumerate all (exponentially many) possible \clName{} boundaries,
but yields fixed-parameter tractability for the combined parameter~$(\ell, |\Sigma|)$.

\begin{Theorem}
  \label{thm:2l-cost1-FPT}
  \klpCoCluProb{2}{*}{\infty} is $O(|\Sigma|^{3\ell}n^2m^2)$-time solvable
  for \mbox{$c=1$.}
\end{Theorem}
\noindent In the following, we prove \autoref{thm:2l-cost1-FPT} in several steps.

A first subresult for the proof of \autoref{thm:2l-cost1-FPT}
is the following lemma, which we use to solve the case where
the number~$2^m$ of possible row partitions is less than~$|\Sigma|^\ell$.
\begin{Lemma}
\label{cor:FPT-mklc}
  For a fixed row partition~$\I$, one can solve \LpCoCluProb[\infty] in $O(|\Sigma|^{k\ell} mn\ell)$ time.
  Moreover, \LpCoCluProb[\infty] is fixed-parameter tractable with
  respect to the combined parameter ($m,k,\ell,c$).
\end{Lemma}
\begin{proof}
  Given a fixed row partition $\I$, the algorithm enumerates all $|\Sigma|^{k\ell}$ different \clName{} boundaries~$\mathcal U=(u_{rs})$.
  We say that a given column~$j$ \emph{fits} in column block~$J_s$ if, for
  each $r\in [k]$ and $i\in I_r$, we have $a_{ij} \in [u_{rs},u_{rs}+c]$
  (this can be decided in $O(m)$ time for any pair $(j, s)$).
  The input is a yes-instance if and only if for some \clName{} boundary
  $\mathcal U$, every column fits in at least one column block.

  Fixed-parameter tractability with respect to~$(m,k,\ell,c)$ is
  obtained from two simple further observations.
  First, all possible row partitions can be enumerated in~$O(k^m)$ time.
  Second, since each of the $k\ell$ \clName{}s contains at most~$c+1$
  different values, the alphabet size~$|\Sigma|$ for yes-instances is
  upper-bounded by~$(c+1)k\ell$.
\end{proof}
%

The following lemma, also used for the proof of \autoref{thm:2l-cost1-FPT}, 
yields that even for the most difficult instances, there is no need 
to consider more than two column clusters to which any column can be assigned.

\begin{Lemma} \label{lem:choiceAmong2}
Let $I=(\matrixI\in\Sigma^{m\times n},k=2,\ell,c=1)$ be an instance of \klpCoCluProb{2}{*}{\infty}, 
 $h_1$~be an integer, $0<h_1<m$, 
and	
  $\mathcal U= (u_{rs})$ be a \clName{} boundary with pairwise different columns such that
	$|u_{1s}-u_{2s}|=1$ for all $s\in[\ell]$.
		
	Then, for any column 
	$j\in [n]$, two indices ${s_{j,1}}$ and ${s_{j,2}}$ can be computed in time $O(mn)$, such 
	that if $I$ has a solution~$(\{I_1,I_2\},\{J_1,\ldots,J_\ell\})$ satisfying $\mathcal U$ with $|I_1|=h_1$, 
	then it has one where each column $j$
	is assigned to either $J_{s_{j,1}}$ or $J_{s_{j,2}}$.
\end{Lemma}
\begin{proof}
  We write $h_2=m-h_1$ ($h_2=|I_2|>0$ for any solution with $h_1=|I_1|$).
  Given a column $j\in[n]$ and any element $a\in\Sigma$, we write
  $\occ{a}{j}$ for the number of entries with value $a$ in column~$j$.

  Consider a column block $J_s\subseteq[n]$, $s\in[\ell]$. Write $\alpha,\beta,\gamma$ for the three values such that
  $U_{1s}\setminus U_{2s}=\{\alpha\}$, $U_{1s}\cap U_{2s}=\{\beta\}$ and   $U_{2s}\setminus U_{1s}=\{\gamma\}$.
  Note that $\{\alpha,\beta,\gamma\}=\{\beta-1,\beta,\beta+1\}$.
  We say that column $j$ \emph{fits} into column block $J_s$
  if the following three conditions hold:
  \begin{compactenum}[(i)]
  \item\label{cond_i}  $\occ{x}{j}=0$ for any $x\notin \{\alpha,\beta,\gamma\}$,
  \item\label{cond_ii}  $\occ{\alpha}{j} \leq h_1$, and
  \item\label{cond_iii} $\occ{\gamma}{j}\leq h_2$.
  \end{compactenum}
  Note that if Condition~(\ref{cond_i}) is violated, then the column contains an
  element which is neither in~$U_{1s}$ nor in~$U_{2s}$.
  If Condition~(\ref{cond_ii}) (respectively~(\ref{cond_iii})) is violated, then there are more
  than~$h_1$ (respectively~$h_2$) rows that have to be in row block~$I_1$
  (respectively~$I_2$).
  Thus, if $j$ does not fit into a column block $J_s$, then there is no solution where $j\in J_s$.
  We now need to find out, for each column, to which fitting column
  blocks it should be assigned.
	
  Intuitively, we now prove that in most cases a column has at most two fitting column blocks, and, in the remaining cases, 
	at most two pairs of ``equivalent'' column blocks.
	
  Consider a given column $j\in[n]$. Write $a=\min\{a_{ij}\mid i\in [m]\}$ and  $b=\max\{a_{ij}\mid i\in [m]\}$.
  If $b\geq a+3$, then Condition~(\ref{cond_i}) is always violated: $j$~does not fit into any column block, and the instance is a no-instance.
  If $b=a+2$, then, again by Condition~(\ref{cond_i}), $j$~can only fit into a column block where $\{u_{1s},u_{2s}\}=\{a,a+1\}$.
  There are at most two such column blocks: we write~$s_{j,1}$ and~$s_{j,2}$ for their indices 
	($s_{j,1}=s_{j,2}$ if a single column block fits).
  The other easy case is when $b=a$, i.e., all values in column $j$ are equal to $a$.
  If $j$ fits into column block $J_s$, then, with Conditions~(\ref{cond_ii}) and~(\ref{cond_iii}), $a\in U_{1s}\cap U_{2s}$, and $J_s$ is one of the at most two
  column blocks having $\beta=a$: again, we write $s_{j,1}$ and $s_{j,2}$ for their indices.

  Finally, consider a column $j$ with $b=a+1$, and let $s\in[\ell]$ be
  such that~$j$~fits into~$J_s$.
  Then, by Condition~(\ref{cond_i}),  the ``middle-value'' for column block $J_s$
  is $\beta\in\{a,b\}$.
  The pair $(u_{1s},u_{2s})$ must be from $\{(a-1,a),
  (a,a-1), (a,b),(b,a)\}$. We write $J_{s_1}, \ldots ,J_{s_4}$ for the
  four column blocks (if they exist) corresponding to these four
  cases.	
	We define $s_{j,1}=s_1$  if $j$ fits into $J_{s_1}$, and $s_{j,1}=s_3$ otherwise.
	Similarly, we define $s_{j,2}=s_2$  if $j$ fits into $J_{s_2}$, and $s_{j,2}=s_4$ otherwise.
	
  Consider a solution assigning $j$ to $s^*\in\{s_1,s_3\}$, with $s^*\neq s_{j,1}$. Since $j$
	must fit into $J_{s^*}$, the only possibility is that $s^*=s_3$ and $s_{j,1}=s_1$.	
	Thus, $j$ fits into both $J_{s_1}$ and $J_{s_3}$, 	
  so Conditions~(\ref{cond_ii}) and~(\ref{cond_iii}) imply $\occ aj\leq h_1$ and  $\occ bj \leq h_2$. Since $\occ
  aj+\occ bj = h_1+ h_2=m$, we have $\occ aj=h_1$ and  $\occ
  bj=h_2$. Thus, placing $j$ in either column block yields the same
  row partition, namely $I_1=\{i\mid a_{ij}=a\}$ and $I_2=\{i\mid
  a_{ij}=b\}$. Hence, the solution assigning $j$ to~$J_{s_3}$, 
	can assign it to $J_{s_1}=J_{s_{j,1}}$ instead without any further need for 
	modification.	
		
	Similarly with $s_2$ and $s_4$, any solution assigning $j$ to $J_{s_2}$ or $J_{s_4}$
	can assign it to $J_{s_{j,2}}$ without any other modification. Thus, 
	since any solution must assign $j$ to one of $\{ J_{s_1}, \ldots , J_{s_4}\}$,
	it can assign it to one of $\{J_{s_{j,1}},J_{s_{j,2}}\}$ instead.
  \end{proof}

We now give the proof of~\autoref{thm:2l-cost1-FPT}.

\begin{proof}
  Let~$I=(\matrixI\in\Sigma^{m\times n},k=2,\ell,c=1)$ be a
  \klpCoCluProb{2}{*}{\infty} instance.
  The proof is by induction on~$\ell$.
  For $\ell=1$, the problem is solvable in~$O(n(m+\log n))$ time
  (\autoref{thm:1l-poly}).
  We now consider general values of $\ell$.
  Note that if~$\ell$ is large compared to $m$ (that is, $2^m < |\Sigma|^\ell$),
  then one can directly guess the row partition and run the algorithm of
  \autoref{cor:FPT-mklc}.
  Thus, for the running time  bound, we now assume that $\ell<m$.
  By \autoref{obs:integer-values} we can assume that~$\Sigma\subset \mathbb Z$.

  Given a \CoClu{2,\ell}~$(\I=\{\{1\},\{2\}\},\J)$, a \clName{} boundary $\mathcal
  U=(u_{rs})$ satisfied by~$(\I,\J)$, and $U_{rs}=[u_{rs},u_{rs}+c]$, each column
  block~$J_s\in \J$ is said to be
  \begin{compactitem}
  \item with \emph{equal} bounds if $U_{1s}=U_{2s}$,
  \item with \emph{non-overlapping} bounds if $U_{1s}\cap U_{2s}=\emptyset$,
  \item with \emph{properly overlapping} bounds otherwise.
  \end{compactitem}

  \noindent We first show that instances implying a solution containing at least
  one column block with equal or non-overlapping bounds can easily be dealt with.

\begin{Claim}\label{claim:eqbounds}
If the solution contains a column-block with equal bounds, then
it can be computed in~$O(|\Sigma|^{2\ell}n^2m^2)$ time.
\end{Claim}
\begin{proof}
  Assuming without loss of generality, that the last column block,~$J_\ell$, 
	has equal bounds.
  We try all possible values of~$u=u_{1\ell}$.
  Note that column block~$J_\ell$ imposes no restrictions on the row
  partition. Hence, it can be determined independently of the rest of
  the co-clustering.
  More precisely, any column with all values in
  $U_{1\ell}=U_{2\ell}=[u, u+c]$ can be put into this block,
  and all other columns have to end up in the~$\ell-1$ other blocks,
  thus forming an instance of \klpCoCluProb{2}{\ell-1}{\infty}.
  By induction each of these cases can be tested in
  $O(|\Sigma|^{2(\ell-1)}n^2m(\ell-1))$ time.
  Since we test all values of~$u$, this procedure finds a solution with a column
  block having equal bounds in
  $O(|\Sigma|\cdot|\Sigma|^{2(\ell-1)}n^2m(\ell-1))=O(|\Sigma|^{2\ell}n^2m^2)$ time.
	\end{proof}
\begin{Claim}
If the solution contains a (non-empty) column-block with non-overlapping bounds, then
it can be computed in~$O(|\Sigma|^{2\ell}n^2m^2)$ time.
\end{Claim}
\begin{proof}
  Write $s$ for the index of the column block $J_s$ with
  non-overlapping bounds, and assume that, without loss of generality,
  $u_{1s}+c<u_{2s}$.
  We try all possible values of~$u=u_{2s}$, and we examine each
  column $j\in [n]$. We remark that the row partition is entirely
  determined by column~$j$ if it belongs to column block~$J_s$.
  That is, if $j\in J_s$, then $I_1=\{i\mid a_{ij}< u\}$ and $I_2=\{i\mid a_{ij}\geq u\}$.
  Using the algorithm described in \autoref{cor:FPT-mklc}, we
  deduce the column partition in $O(|\Sigma|^{2\ell-1}nm\ell)$ time,
	which is bounded by $O(|\Sigma|^{2\ell}n^2m^2)$.
\end{proof}

  We can now safely assume that the solution contains only  column
  blocks with properly overlapping bounds.
  In a first step, we guess the values of the \clName{} boundary
  $\mathcal U= (u_{rs})$. Note that, for each~$s\in[\ell]$, we only
  need to consider the cases where
  $0<|u_{1s}-u_{2s}|\leq c$, that is, for $c=1$, we have $u_{2s}=u_{1s}\pm 1$.
  Note also that, for any two distinct column blocks $J_s$ and
  $J_{s'}$, we have $u_{1s}\neq u_{1,s'}$ or  $u_{2s}\neq u_{2,s'}$.
  We then enumerate all possible values of $h_1=|I_1|>0$ (the
  \emph{height} of the first row block), and we write
  $h_2=m-h_1=|I_2|>0$.
  Overall, there are at most $(2|\Sigma|)^\ell m$ cases to consider.
	
	Using \autoref{lem:choiceAmong2}, we compute integers $s_{j,1}, s_{j,2}$ for each 
	column $j$ such that any solution satisfying the above conditions (\clName{} boundary
	$\mathcal U$ and $|I_1|=h_1$) can be assumed to assign each column $j$ to one of $J_{s_{j,1}}$ or $J_{s_{j,2}}$.
		
	We now introduce a \textsc{2\dash{}Sat} formula allowing us to simultaneously
  assign the rows and columns to the possible blocks.
  Let~$\phi_{\matrixI,\mathcal U}$ be the formula as provided by \autoref{lem:reduc-to-sat}.
  Create a formula $\phi'$ from $\phi_{\matrixI,\mathcal U}$ where,
  for each column $j\in [n]$, the column clause $C_j$ is replaced by
  the smaller clause $C'_j:=(y_{j,s_{j,1}} \vee y_{j,s_{j,2}})$.
  Note that $\phi'$ is a \textsc{2\dash{}Sat} formula since all other clauses $R_i$
  or~$B_{ijrs}$ already contain at most two literals.

  If $\phi'$ is satisfiable, then $\phi_{\matrixI,\mathcal U}$ is satisfiable
  and~$\matrixI$ admits a \CoClu{2,\ell} satisfying~$\mathcal U$.
  Conversely, if~$\matrixI$ admits a \CoClu{2,\ell} satisfying~$\mathcal U$
  with~$|I_1|=h_1$, then, by the discussion above, there exists a
  co-clustering where each column~$j$ is in one of the column
  blocks~$J_{s_{j,1}}$ or~$J_{s_{j,2}}$.
  In the corresponding Boolean assignment, each clause of
  $\phi_{\matrixI,\mathcal U}$ is satisfied and each new column
  clause of~$\phi'$ is also satisfied. Hence, $\phi'$ is satisfiable.
  Overall, for each \clName{} boundary~$\mathcal U$ and each~$h_1$, we
  construct and solve the formula~$\phi'$ defined above.
  The matrix~$\matrixI$ admits a \CoClu{2,\ell} of cost~1 if and only if $\phi'$ is
  satisfiable for some~$\mathcal U$ and~$h_1$.

  The running time for constructing and solving the formula~$\phi'$, for any fixed \clName{}
  boundary~$\mathcal U$ and any height~$h_1\in[m]$, is in~$O(nm)$,
  which gives a running time of $O((2|\Sigma|)^{\ell}nm^2)$ for this last part.
  Overall, the running time is thus $O(|\Sigma|^{2\ell}n^2m^2 +
  |\Sigma|^{2\ell}n^2m^2 + (2|\Sigma|)^{\ell}nm^2) = O(|\Sigma|^{2\ell}n^2m^2)$.
\end{proof}
Finally, we obtain the following simple corollary.

\begin{Corollary}
\label{cor:2l-FPT}
  \klpCoCluProb{2}{*}{\infty} with~$c=1$ is fixed-parameter tractable with
  respect to parameter~$|\Sigma|$ and with respect to parameter~$\ell$.
\end{Corollary}
\begin{proof}
  \autoref{thm:2l-cost1-FPT} presents an FPT-algorithm with respect to the
  combined parameter $(|\Sigma|,\ell)$.
  For \klpCoCluProb{2}{*}{\infty} with~$c=1$,
  both parameters can be polynomially upper-bounded within each other. Indeed, $\ell<|\Sigma|^2$
  (otherwise there are two column blocks with identical \clName{} boundaries,
  which could be merged) and $|\Sigma|<2(c+1)\ell=4\ell$
  (each column block may contain two intervals, each covering at
  most $c+1$~elements).
\end{proof}

\section{Conclusion}

Contrasting previous theoretical work on polynomial-time approximation algorithms
\cite{ADK12, JSB09}, we started to closely investigate
the time complexity of exactly solving the NP-hard \LpCoCluProb[\infty]
problem, contributing a detailed view on its computational
complexity landscape.
Refer to \autoref{tab:results} for an overview on most of our results.

Several open questions derive from our work.
Perhaps the most pressing open question is whether \klpCoCluProb{2}{*}{\infty} is polynomial-time solvable for larger constant-size alphabets.
So far, we only know that~\klpCoCluProb{2}{*}{\infty} is
linear-time solvable for ternary matrices (\autoref{thm:sigma-3-poly})
and quadratic-time solvable for every constant-size alphabet if~$c=1$ (\autoref{cor:2l-FPT}).
Another open question is the computational complexity of higher-dimensional
co-clustering versions, e.g. on three-dimensional tensors as input (the most basic case here corresponds to (2,2,2)-\LpCoCluProb[\infty], that is, partitioning each dimension into two subsets).
Indeed, other than the techniques for deriving
approximation algorithms \cite{ADK12, JSB09}, our exact methods
do not seem to generalize to higher dimensions.
Last but not least, we do not know whether \textsc{Consecutive}
\LpCoCluProb[\infty] is fixed-parameter tractable or~W[1]-hard with respect to the
combined parameter~$(k,\ell)$.

We conclude with the following more abstract vision on future research:
Note that for the maximum norm, the cost value~$c$ defines a ``conflict
relation'' on the values occurring in the input matrix. That is, for
any two numbers $\sigma,\sigma' \in \Sigma$ with  $|\sigma-\sigma'| > c$, we know that
they must end up in different \clName{}s. These conflict pairs completely
determine all constraints of a solution since all other pairs can be
grouped arbitrarily.
This observation can be generalized to a graph model.
Given a ``conflict relation'' $R\subseteq \binom{\Sigma}{2}$
determining which pairs are not allowed to be put together into a \clName{},
we can define the ``conflict graph''~$(\Sigma, R)$.
Studying co-clusterings in the context of such
conflict graphs and their structural properties
could be a promising and fruitful direction for future research.

\paragraph*{Acknowledgments.}
We thank Stéphane Vialette (Université Paris-Est Marne-la-Vallée) for stimulating discussions.

\bibliography{co-cluster}
\bibliographystyle{abbrvnat}

\end{document}